\documentclass[11pt,twocolumn,a4paper]{article}

\usepackage{graphicx,amsmath,amssymb,amsthm,mathabx,booktabs,txfonts}
\date{}

\author{Yong Tan\\~\\
\emph{entermarket@163.com}}
\title{Study Morphology of Minimum Spanning Tree Problem and Generalized Algorithms}

\theoremstyle{plain}
\newtheorem{theorem}{Theorem}
\newtheorem{lemma}{Lemma}

\newtheorem{definition}{Definition}

\begin{document}
\maketitle 
\begin{abstract}
In this paper, we study the form over the \emph{minimum spanning tree problem} (MST) from which we will derive an intuitively generalized model and new methods with the upper bound of runtimes of \emph{logarithm}. The new pattern we made has taken successful to better equilibrium the benefits of local and global when we employ the strategy of \emph{divide and conquer} to optimize solutions on problem. Under new model, we let the course of clustering become more transparent with many details, so that the whole solution may be featured of much reasonable, flexibility, efficiency and approach to reveal or reflect the reality. There are some important methods and avenues as fruits derived from discussions or trial which can be broad usefulness in the fields of graphic analysis, data mining, \emph{k-means} clustering problem and so forth.
\end{abstract}

~\newline
\section{Introduction}
MST problem has been disposed to have a broad application in practice, such as communication, network problem, biology, clustering, data mining and decision problem else. The problem could be described as: \textquotedblleft given an instance, every edge on it involves of a weight which can be a number or a vector. MST is such a tree that concurrently is of a cutting graph out of original ones. On MST, there are minimum total weight and maximum number of endpoints that means including all nodes for exact one.\textquotedblright 

Meanwhile it is a problem how to acquire the optimum decision that presents a course of aggregation of either edges or nodes. In \cite{2}, MST problem were introduced to a special instance of GMST (generalized minimum spanning tree) problem that has had being a \emph{NP-hard} problem so far. Anyway, peoples all and all at present like to follow a classic way to solve this kind of problems, which course always starts at someone as a \emph{seed} and incident to a set $G^{\prime}$ created of singleton automatically. Thereafter the singleton $G^{\prime}$ sets out to absorb the other nodes unceasingly which is in the complement $G\setminus G^{\prime}$. The course of merger will continue until $G^{\prime}$ absorbs the final element.

When we use a linear algorithm to solve this kind problem, we may likely meet some hardship. For example, underneath theoretical pattern of GMST, with a small input, we may suffer huge scale computing out of probability of combination or permutation. That case would request our resource of computation much and readily lead to exceed our capacity. 

Hereby peoples are disposed to employ a strategy we call \emph{divide-and-conquer}. The core idea of this measure is of partition a big scale (global) problem to several sub-problem (local) and, in order to get the overall optimum through per local optimization on each part; for a typical example \emph{dynamical program} by using recursive operator. But meanwhile this methodology has to meet another new problem, even is so serious that how to cope with the benefits on \emph{local} and \emph{global}. Sometimes there is a big distinct among what our earning and reality, simply speaking, when we figure the sum of results out of several parts, this process may not surely be a simple calculation of adding them together. Perhaps is a polynomial or even to be a mistake at beginning because of short of necessary context that we have not known. Certainly our solution could be likely conducted to lose some information or relationships among those divided parts within computing process. Then what strategy we need becomes vital important to us, which can guarantee the ultimate fruit is correct. 

Hence, studying some natural-born properties on problem oneself to gain a better and generic strategy manifests a critical significance to our resolution.

~\newline
\textbf{Related Works.} When we talk about MST, we must have to refer to two traditional and typical algorithms \emph{Kruskal}\textquoteright s algorithm\cite{3} and \emph{Prim}\textquoteright s algorithm\cite{4}. They are distinct at solvable course, remarkable of: the fore one is there a pretreatment on \emph{raw} data, where \emph{Kruskal} sort those weights in order of ascent which is incident to an edge set. The sequent work is further of picking up edges into $G^{\prime}$, naturally as the order of ascent on their incident weights as \emph{greedy method} via \emph{lesser-first}. The later one randomly chooses an endpoint as seed into $G^{\prime}$, of course a singleton being created thereof. The sequence is set $G^{\prime}$ continues to augment by absorbing endpoints one by one till to include all of nodes. In this form, for any candidate of edge such that one end has been taken into set $G^{\prime}$ but another outside. \emph{Prim} will compare those outer suchlike in order to enumerate a minimum one among them as an eligible endpoint to absorb. 

In sum, Kruskal\textquoteright s measure needs sorting all weights and must to settle their incident edges, so the potential complexity of such pretreatment may at most implicate to \emph{square} of input, as though the later subroutine of merger could be conceded as a logarithm $O(m\log n)$ over runtime where $m,n$ are the amounts of edges and endpoints respectively. The Prims\textquoteright thus must suffer $n$th bouts of perhaps about $m$ times of enumerating, testing and receiving so that runtimes has still got indirect to $O(mn)$ at worst case. 

Otherwise, they similarly meet some potential cost more and less that may lead to consume more resources in general case: when many candidates of edge are incident to a same number simultaneously, the two methods have must to pay for a higher expenditure to deal with the benefit between local and global to ensure the local selection would not impact the correctness of overall result. Thus we can see their claims in papers always emphasis: distinct edge associated to distinct weight.

On the other hand, there is a shortage to such two measures that there are two sets \emph{selected} $G^{\prime}$ and \emph{unselected} $G\setminus G^{\prime}$ throughout the whole course of merger, so that it is weighted of heavy scarcely to present a conglomeration of many clusters. In \cite{2}, many researchers attempted to attribute those kinds of problems to a typical process of clustering, where they even said there exists such initial situation that those prime nodes all are incident to respective singleton. Hence, as to MST problem, they thought about what was just a special case in GMST. \cite{2}

As to others efficient methods\cite{8, 9, 10}, their devotions and works were disposed to reform the procedure or process for example Harold N. Gabow et al\cite{5} used the \emph{Fibonacci heap} to speed up the execution by accessing minimum weight and incident edge at $O(1)$ to improve the runtime, and alleged to obtain the optimum of $\log^{(i)}=\log\log\ldots.\log$ by this way. But the instance must satisfy something. 

In\cite{6}, Karger et al \emph{with high probability in the restricted random-access model} attempted to exploit the probability of sampling unavailable edge to accelerate the process of cutting original graph till to a MST which can be said \emph{whole-to-part} style, which is difference to those precedents at an utterly inverse strike. 

Today in this paper, we find another way that to survey the morphology about this problem. To our algorithms, they would be abstracted to an order of first \emph{whole} to \emph{local} and second \emph{bottom} to \emph{top}. Otherwise to convenience for exposition, we will divide the whole course into two steps: the first is \emph{node} stage that there is a task of which all of nodes will have been set in $K$ clusters. The value of $K$ is not given by us in advance but inside instance possessed of inherent characters or relationships among those nodes. Then after is \emph{cluster} stage in which clusters will merge one another until an utter MST appears.

~\newline
\emph{Our paper is outlined} as: the next section is \emph{preliminary} to introduce background knowledge, terminology and some claims in routine. The third section will contain definitions and proofs about the form of MST to achieve the task of \emph{whole} to \emph{local}. The new algorithms will be expounded with \emph{pseudo} code. And they will be discussed of over those respects of complexity, optimization, theoretical proofs, trial and some fruits else. The final section will give some conclusions.

\section{Preliminary}
When we refer to a graph, the term $G=(V, \tau)$ is a basic form used by us. The instance $G$ is consisted of node set $V$, arc set $\tau$. According to convention, we let $n=\vert V\vert$ which variable $n$ is of cardinal of set $V$. It is reasonable that we allow all nodes on instance maps to a natural number queue: $1,2,3,\cdots, n$ as their alias. Then we can use an index array to present such data structure for easy to access.

An \emph{arc} between a pair of nodes $u,v\in V$ usually is characterized by an order pair of $(u,v)$ to express a direction from $u$ to $v$, unless otherwise definition given. Another traditional notion of \emph{edge}, in fact, it contains a pair of arcs both strikes is one another opposite. 

Given a group of numbers $W$ and we allow each member $\omega\in W$ is at least incident to an edge $e$ on instance, we call such graph \emph{weighted-edge graph}. The weight $\omega$ actually is also in concurrent effect to two arcs contained in edge $e$. Otherwise we claim that every weight is \emph{positive} greater than zero, which is permitted of our consensus about MST problem.

There is a particular data structure \emph{united subgraph}\cite{1} $s_{i}= R(i)\times L(i)$ for $u\in R(i), v\in L(i)$ such that $(u,v)\in\tau$. If we view set $\tau$ as a family, then the united subgraph $s_{i}$ would be the \emph{set partition}\cite{1} on $\tau$, which presents a Cartesian product of $R(i), L(i)$; moreover set $R(i)$ is a singleton. Therefore 
this special form can be described of a \emph{star-tree} where many arcs are consisted of unique \emph{root} in root set $R(i)$ and many leaves in leaf set $L(i)$, and to show a root onto many \emph{leaves}. We let $ m=\vert L(i)\vert $ where $m$ represents the size of a leaf set $L(i)$. Then we let $E=nm$ as an approach formula to cast the many of arcs on instance. 

Finally, we can define a total weight as: for a path or tree $P=(e_{1}, e_{2},\ldots,e_{k})$, the total weight $W_{P}$ is the sum of weights incident to edges that has been in path $P$: $W_{P}=\sum_{i=1}^{k}\omega_{e_i}$ for $\forall e_i\in P$. 

\section{Morphology}
Based on our consensus at MST problem, given an instance should be ascribed to a \emph{simple graph} where the every connected relationship between nodes $u,v$ must be followed to an edge $e$ with dual channels connecting two ends. The traditional concept \emph{undirected} in effect is dualistic strikes sharing a weight so that we can utilize the united subgraph structure to characterize the local benefit done for every node that as a root taken onto many leaves, unless isolated ones. In\cite{5} Gabow used this measure alike to lift up the capacity of obtaining edge and incident weight which is in Fibonacci-heap constructed of arcs and their incident weights.

Considers the data structure of \emph{united subgraph} formatted by $s= R\times L$, we may describe a relation among subgraph $s$ and its incident weight set $W^{\prime}$. Further there may be a mapping function $C$ and its action as 
\begin{flalign}\label{d1}
C_{s}\colon \alpha\rightarrow \omega\vert&~ \omega\in W^{\prime}\text{ and }\alpha\in s&\\
&\text{s.t. }\vert s\vert\geq\vert W^{\prime}\vert\text{ and }W^{\prime}\subseteq W.&\nonumber
\end{flalign}
We call set $W^{\prime}$ \emph{local weight set}. We further sort these numbers in local weight set $W^{\prime}$, as such we can have some definitions as follows.

\subsection{Definitions}
Follows the above-mentioned function\eqref{d1}, we may produce a triple array $M$ consisted of nodes and weight as follow.
\begin{flalign}\label{eq1}
M=&(r,\ell,\omega)\\
\text{s.t.}&\quad\nonumber\text{(1) }\forall s_{i}\subseteq\tau \Rightarrow r\in R_{i} \text{ and }\ell\in L_{i}\\
&\quad\text{(2) }\omega\in W_{i}\text{ and }W_{i}\subseteq W.\nonumber
\end{flalign}
~\newline
To the term\eqref{eq1}, we call it \emph{Arc Weight Triple} abbr. AWT. It is obvious for such logical structure that we can take subgraph $s$ and incident local weight set $W^{\prime}$ to consist a group of AWTs, so that these AWTs share a same root $r$. 

Hence in set $W^{\prime}$, while we seek out a special kind of AWT who does with minimum weight thereof, we may describe a \emph{local benefit} for node $r$, and we can as such treat other nodes. 

Since that, we have the first definition about the local benefit at node $r$ as follows.

\begin{definition}\label{d1}
Consider a AWT $(r,\ell,\omega)$ with $\omega$ is the minimum member in set $W_{i}$, we call $r$ be \emph{subject} into leaf $\ell$, denoted by $r\rightharpoonup\ell$ and, we call weight $\omega$ \emph{Minimal Vassal Cost}, abbr. \emph{MVC}. Meanwhile we also denote MVC by $\omega_{r\ltimes\ell}$, on the other hand we use the term $\omega_{r}^{\ltimes}$ to represent a pure quantity of MVC not to shoot any specific leaf.
\end{definition}

This definition is readily taken to comprehend: we can suppose the leaf $\ell$ has been in some entity. When the entity plans to absorb root $r$, by the edge $\{r,\ell\}$ between them, the number of MVC $\omega_{r}^{\ltimes}$ could be referred to a minimum cost as input for some computation. But swap $r$ and $\ell$ positions in the fore process, which is from root $r$ to absorb leaf $\ell$, the situation may be not as same as the fore case. Because there may be $\omega_{\ell}^{\ltimes}<\omega_{r}^{\ltimes}$, it seemingly does not to satisfy the condition of by minimum cost. We call this asymmetric course \emph{invert pitfall}.

Since a AWT consisting of root $r$ and leaf $\ell$ concurrently, then there may be a case of $\omega_{r\ltimes\ell}$ and $\omega_{\ell\ltimes r}$ existing concurrently. It is apparent for the edge $\{r,\ell\}$ to be a minimal bridge between two nodes. We define such structure as well as follows.

\begin{definition}\label{d2}
Following the definition\eqref{d1}, consider $\exists\omega_{r\ltimes\ell}\in W_{r}$ and $\exists\omega_{\ell\ltimes r}\in W_{\ell}$ on instance, we denote them by $r\rightleftharpoons \ell$ or $\ell\rightleftharpoons r$. And say such abreast relation refers to a \emph{beam structure}.
\end{definition}

Now we use MVC to construct an abstracted framework that characterized as to the so-called local benefit, and take it to bond with the based structure of united subgraph. Hereby we call this new pattern \emph{fleet model}. 

Of course, the fleet model actually is a cutting graph abstracted out of original instance as well, we denote it by $\mathcal{M}$.

On some intuitive sense, we can readily give some features around this new pattern as follows, which are easy to proven and we forbear to go into details for them.
\begin{enumerate}
\item The new relation $\mathcal{M}$ should cover every node $r$ on instance $G$ if and only if node $r$ is not isolated without any neighbor. That said the leaf set $L_{r}\neq\varnothing$.

\item Consider a MVC $\omega_{r\ltimes\ell}$, we have $\omega_{r}^{\ltimes}\geq\omega_{\ell}^{\ltimes}$. 

\item Consider a triple nodes $s,u,t$ satisfies the form $s\rightharpoonup u, u\rightharpoonup t$, such that there is a due \emph{transitive} relation $\omega_{s}^{\ltimes}\geq\omega_{u}^{\ltimes}\geq\omega_{t}^{\ltimes}$.

\item Consider a MVC $\omega_{r\ltimes\ell}$, it is slightly to prove if there is an entity on instance to merge with node $r$, the cost is cheapest to pass through the arc $(\ell,r)$ to absorb. Of course, this feature is guaranteed to hold by the context of that sum of several weights must be monotonic of increment: $P^{\prime}\subseteq P\Rightarrow W(P^{\prime})\leq W(P)$. 

\item We call a group consisting of nodes and MVC as \emph{flotilla}, anyway the minimum flotilla just can be a \emph{beam} with a couple of nodes.
\end{enumerate}
~\newline
We have finished the prime stage work of constructing basic framework. This abstract comes out of instance which can be said of \emph{global to local} alike. That cutting graph $\mathcal{M}$ is possessed of mixed direction with both styles of \emph{directed} and \emph{bidirected}. In fleet model, it is secondary for us to concern if there is any cycle on graph $\mathcal{M}$. This case is completely distinct with method of MSF\cite{8}, which was summaried over $Bor\mathring{u}vka$\textquoteright s method by Seth Pettie et al \cite{8}.

In addition, graph $\mathcal{M}$ is not thus to said as a MST that just a model described of a total local benefit. On the other hand, by the $5^{\text{th}}$ feature above, graph $\mathcal{M}$ can be a group of many flotillas where those flotillas are separate one another. As to any flotilla, the MSF is just a subset of a flotilla, so that we need chase those strings on it. 

Then we need first to deal with the problem of invert pitfall: a node $\ell$ can absorb node $r$, but how about to settle the node $\ell$; what is about to happen on it? These questions may be taken to involve with the concrete form of those flotillas within the range of graph theory as like as geometry that is an intuitive model. That will be the contents of our research in next session.

\subsection{Proof}
In order to facilitate describing those relations among nodes in pertinence and intuition, we call a node \emph{boat} or \emph{towboat}, well is it clear for term $r\rightharpoonup \ell$ that hints the leaf $\ell$ is a towboat hauling the fore boat $r$. Certainly they both concurrently may be towboats to each other if they are two ends of a same beam. Then a \emph{flotilla} is an intricate group organized by boats (nodes), \emph{ropes} (arcs) and strains (weights). We reserve traditional word \emph{cluster} to express the meaning of a collection of pure nodes. 

Consider a flotilla $F=b_1,b_2,\ldots,b_{\kappa}$, we can suppose there is a sequence
\begin{flalign*}
\Gamma=b_1\rightharpoonup b_2\rightharpoonup\ldots\rightharpoonup b_{(\kappa - 1)}&\rightharpoonup b_{\kappa}&\\
&\text{ for }1\leq\kappa\leq\infty&
\end{flalign*}
By the above $3^{\text{rd}}$ feature, we can have a sequence with comparative manner among them
\begin{flushleft}
$\Omega = \omega_{1 }^{\ltimes} \geq \omega_{2}^{\ltimes}\geq\ldots\geq\omega_{(\kappa - 1)}^{\ltimes}\geq\omega_{\kappa}^{\ltimes}.$
\end{flushleft}
And it is clear that those member in sequence $\Omega$ may be mapped to a queue of discrete scalars 
\begin{flushleft}
$W^{\/*}=x_1, x_2,\ldots,x_{j}$\quad for $j\leq\kappa $. 
\end{flushleft}
When $\vert W^{\/*}\vert = 1$, it knows for $\forall\omega_{i}^{\ltimes},\omega_{j}^{\ltimes}\in\Omega (i\neq j)$ such that $\omega_{i}^{\ltimes}= \omega_{j}^{\ltimes}$. Then we have a result that flotilla $F$ is a complete group of beams and the sequence\textquoteright s length $\vert \Gamma\vert$ can be \emph{infinite}. 

Likewise, an infinite beam group can be sited at the medium position of the queue $\Omega$, but the series $W^{\/*}$ still may be a finite set with $\vert W^{\/*}\vert\leq \kappa$. 

Besides those discussions about infinite set, we only treat those sets of finite and the queue in $\Omega$ should stop at $b_{\kappa}$ as a strictly order form
\begin{flushleft}$ \Omega = \omega_{1 }^{\ltimes} > \omega_{2}^{\ltimes}>\ldots>\omega_{(\kappa - 1)}^{\ltimes}>\omega_{\kappa}^{\ltimes}.$\end{flushleft}
And sequence $W^{\/*}$ as a series has $\vert W^{\/*}\vert= \kappa$. Hereby we can suppose all of above collections that their own cardinals would no longer increase to $\kappa + 1$. Based on the assumption, we will figure out the sequence and further through a proof on a theorem to obtain a conclusion in morphology.

\begin{theorem}\label{t1}
Consider $\Gamma=(\rightharpoonup, b_{i})_{1\leq i\leq\kappa}$ as a part of flotilla and sequence $\Omega= (>,\omega_{i}^{\ltimes})_{1\leq i\leq\kappa}$ is incident to $\Gamma$. If set $\Omega$ is finite then the end of sequence $b_{\kappa}\in\Gamma$ would at least be a member in a beam.
\end{theorem}

\begin{proof}
Follows the above assumption, there is at least the relationship hold for two nature number $1\leq s,t\leq\kappa$ with $ t - s = 1$, such that $\omega_{s\ltimes t}> \omega_{t}^{\ltimes}$. it is inferable to confirm sequence $\Omega$ would converge at $\omega_{\kappa}^{\ltimes}$, if and only if sequences $\Omega, \Gamma$ are finit, otherwise there is always a more less number $\varepsilon < \omega_{\kappa}^{\ltimes}$, then let $\exists (b_{\kappa}, b_{(\kappa+1)}, \omega_{\kappa\ltimes(\kappa+1)})\in G$ and $\omega_{(\kappa+1)}^{\ltimes}= \varepsilon$, then either sequence of $\Gamma$ or $\Omega$ should be \emph{infinite}.

Further consider the ends at both queues $\Omega, \Gamma$, they are $\omega_{\kappa}^{\ltimes}$ and $b_{\kappa}$ respectively; it is deterministic that there is a result of $\omega_{(\kappa - 1)}^{\ltimes} > \omega_{\kappa}^{\ltimes}$. That says there is existence of an \emph{leaf} $\ell$ about $b_{\kappa}$ but not in sequence $\Gamma$, and the leaf $\ell$ subjects $b_{\kappa}$ to itself. So we can have $\omega_{\kappa}^{\ltimes}\geq \omega_{\kappa\ltimes\ell}$ to characterize what happen at the end $b_{\kappa}$: there will be two conclusions, someone is correct for either $\omega_{\kappa}^{\ltimes} = \omega_{\ell }^{\ltimes}$ or $\omega_{\kappa}^{\ltimes}>\omega_{\ell }^{\ltimes}$.

If the later one is correct that is clear to have a more less number $\varepsilon$ to let two sequences cannot be converge at the end $b_{\kappa}$. That so far is contradiction to the initial premise of $\Gamma, \Omega$ as a \emph{finite} set. 

On contrary, although the node $b_{\kappa}$ maybe stays in an infinite structure of beam with sharing a same number, but the case could not affect that prerequisite of sequence $\Omega$ being finite with relation $>$ and converge at $\omega_{\kappa}^{\ltimes}$. 

In turn, the above proof shows that the existence of beam in the flotilla also is the necessary condition for $\vert W^{\/*}\vert < \vert \tau\vert /2$. 

\end{proof}
~\newline
In\cite{8}, Seth and Vijaya had expressed suchlike thought approaching to this theorem\eqref{t1} above: there is a troublesome cycle with a heaviest weight which part has been in MSF $G^{\prime}$. They exclusively mentioned to cutting the heaviest edge which is involved to so-called DJP algorithm made by Jarnik[1930] and rediscovered by Dijkstra[1959] and Prim[1957]\cite{8}. 

We may describe such case by theorem\eqref{t1} model as well, Where are two sequences $\gamma_1$ and $\gamma_1$ to construct a cycle and their two ends are connected in a beam respectively, certainly the either one should be with greater cost which is that heaviest edge concerned by Seth. 

For the lightest beam on cycle, the beam is with minimal cost less than and equal to others. This form is just proven by theorem\eqref{t1}. Consequently, if the track is done along the monotonic string out of heaviest beam, we will steady and surely reach the lightest ones. Hereby we call lightest beam \emph{top} of flotilla, in turn called \emph{bottom} to heaviest. It is obvious that these past algorithms else run around this form for example Kruskal\textquoteright s measure is typical top to bottom.

Follows this form we can construct our methods as \emph{local} to \emph{global} on two phases. The first is similar to $Bor\mathring{u}vka$\textquoteright s method\cite{8}: the node stage in our method is contrasting to \emph{MSF phase} rose by Seth. And the distinct on both methods is that our ones can take nodes out in batch within per loop; no longer need to compare the whole stuffs that have ever been contracted for enumerate only eligible one. On the second stage, the operation becomes to cluster merger, but still we will redo our strategy on those clusters as it has been to nodes. Eventually the procedure would halt because there only one cluster be turned out.

\section{Algorithms}
Firstly follows the form out of theorem\eqref{t1}, we can make a \emph{certificate} function to charge nodes. It is obvious for a top beam in a flotilla, those nodes in it cannot be subjected to others that with greater MVCs. If our strategy of collection is ruled to the law of from \emph{top to bottom} and \emph{peer to peer} in beam, then we could avoid the invert pitfall and the course of reaping is safety too. 

In the implemented course, we surely need to set up a mechanism to prevent MSF from cycle produced by procedure. So we will have some contents to talk about that problem in below session.

\subsection{Pseudo Code and Proof}
We define a function \emph{Charge} to score a node over its quality as follows:
\begin{flushleft}
$Charge(r, \ell)=\left\{
\begin{array}{rl}
1 &\text{ if } \exists r\rightharpoonup\ell \text{ and }\omega_{r}^{\ltimes}>\omega_{\ell}^{\ltimes};\\
2 &\text{ if } \exists\ell\rightharpoonup r\text{ and }\omega_{r}^{\ltimes}<\omega_{\ell}^{\ltimes};\\
3 &\text{ if } \exists r\rightharpoonup\ell \text{ and }\omega_{r}^{\ltimes}=\omega_{\ell}^{\ltimes};
\end{array}\right.$
\end{flushleft}
~\newline
It is easy to realize the meaning of three values to present: 1 says $r$ is a \emph{boat} to leaf $\ell$; 2 is root $r$ as a \emph{towboat} to haul leaf $\ell$; 3 presents they are a \emph{beam}. By the theorem\eqref{t1}, we just pluck these nodes of both two types of score 3 and 3 joined of 2 as the prime point for start survey. Then we have the pseudo code of \emph{node stage} as follows.
\begin{enumerate}
\item Scans every AWT $(r,\ell,\omega)$ consisting of subgraph $s$ and incident local cost set $W^{\prime}$; selects the MVC thereof and record it as $\Omega[r] = \varw_{r}^{\ltimes}$. Meanwhile records the relations about $\rightharpoonup$ on graph $\xi$ that is an empty in initial phase but with same structure as original instance $G$. The formatted storage is $\xi[\ell][i]=r$ where it is only for root $r$ being written into leaf $\ell$\textquoteright s leaf set, and beam is presenting of mutual notch. In this subroutine, we may get a graph $\mathcal{M}=(V,\xi)$.

\item Uses function $Charge$ on all of nodes. In graph $\mathcal{M}=(V,\xi)$, we score root $r$ over \emph{towboat} or \emph{beam}, in turn, score leaf $\ell$ over \emph{boat}. Of course, this behavior is interactive to each other and, we can make a search list $N$ recording of available nodes for start conglomeration over each flotilla.

\item Traversal in list $N$, certainly the chase is also in graph $\mathcal{M}$. With the root and leaf set, the reaping is though root to absorb its own leaves. The process is iterative that these leaves will then become new roots to absorb their own leaves respectively. Our tactic is of from top to bottom or peer to peer and forbid an inverse strike of \emph{down to up}, then more practical execution is carried on by compare those MVCs of root $r$ and leaf $\ell$, and other weights in AWT consisting of root $r$ and leaf $\ell$. 

\item In the reaping course, we use an array $\overline{\xi}$ that holds a same structure as original graph $G$ to record the edge plucked by us, which is formed as united subgraph as $\overline{\xi}[r][i]=\ell$, $\overline{\xi}[\ell][j]=r$. Meanwhile, we use a \emph{n}th index array $C$ store those nodes for several clusters, and responding to prepare an auxiliary array $C^{\/*}$ to store the offsets of clusters in array $C$. Then we gain a group of MSFs $\varg =(C, \overline{\xi})$.

\item In order to manage the node merger course, we devise a gadget to prevent course from repeated absorbing. We use a variable \emph{counter} whom can increase by automatically oneself adding, we employ this variable to produce natural number \emph{id} not only for several clusters, but for as such tag every node bond with some cluster, which we use an index array $\mathcal{I}$ to store these information as $\mathcal{I}[node]=id$. Consequently the program could learn the status about a node through array $\mathcal{I}$, and further by the status to prevent program from repeated receiving and producing cycle in new MSF.
\end{enumerate}
~\newline
\textbf{Discussion.} Above all is runtime \emph{complexity}, it is considerably simple that on the level of encode, our algorithm maybe roughly divide the whole course into three parts which every phase has to scan the whole or partial data, of course including sets $\tau, W$. The complexity naturally is $O(3E)$. 

We call this method \emph{Oriented Abstracted Gradient}, abbr. \emph{OAG} that means we convert those concrete numbers into a fleet model with abstracted relation of \emph{greater, less} and \emph{equal}, finally we achieve the calculation relied on the quality of convergence in sequence.

About the \emph{output} in node stage, we reap a cutting graph $\varg=(C, \overline{\xi})$ and a group of auxiliary variables, including $counter$, $C^{\/*}$ and $\mathcal{I}$. That is to say our fruits resemble \cite{8}\textquoteright s group of MSFs, even to the process of algorithms, both chasing are done along the gradient. But the core thought of ours\textquoteright is remarkable of distinct and similar to \cite{8}:

\begin{enumerate}
\item Our algorithm rears on an arbitrary instance where many distinct edges maybe share with a same weight; contrast to the past methods, ours is much more to be possessed of generalized. Frankly speaking, by theorem\eqref{t1}, that is easier to derive the conclusion there must be no cycle in any flotilla underneath the past context about edge and weight.

\item In \cite{8}, Seth raised to employ \emph{decision tree} to solve a troublesome problem, that in MSF made by Seth, program could not identify the edges in which MSF when they only concern to contract nodes at first stage. Hence, this sequent caused Seth to pay off a large expenditure on computable resource. He had to suggest the opinion about \emph{dense} $m/n$ (here \emph{m} is the amount of edges) must be $\leq\log^{(3)}n$, or else a unknown sequent of vast computing effort could not be afforded which is induced by necessary to entirely traverse a decision tree. In turn, we wield the united subgraph to solve this problem: to master arc equal to manipulate edge. Our fruits out of node stage would guarantee us to rapidly access the information at a lower cost in the next stage, because the form of data structure over a graph is still to keep throughout whole course. 

\item In same way, Seth similar suggested the variable $counter = \log^{(3)}n$, by this settlement to deal with computing decision tree. It is in practice to take a hardship to choose a \emph{proper} instance to us.

\item Taken together, there are only two types for a cycle: one is two paths start at a same endpoint, and both stop at another same point, at the medium where there is not any intersection amid both intervals. Second is such a path where any point on it as a start point also becomes the end point by traversal along the path. Seth\cite{8} exploited two statuses \emph{live} and \emph{dead} to tag node to avoid cyclic happen. We use the id of clusters to indicate status for nodes, actually ours does with the same meaning to Seths\textquoteright.
\end{enumerate}

Now, we have obtained a group $\varg$ of MSFs, which could be controlled by maneuver on array $C^{\/*}$ and $\mathcal{I}$. About the form of our MSF, in fact, the similar proof has been done by Seth. But via different context, we have to do it again by a lemma.

\begin{lemma}\label{l1}
After a method \emph{OAG} executing on a given graph $G$, a component $\varg_i$ in result $\varg=(C,\overline{\xi})$ is a MSF.
\end{lemma} 

\begin{proof}
By theorem\eqref{t1}, we suppose a flotilla with a convergent sequence $P$. Further assume a beam $\mathcal{B}$ as initial entity, then there are three cases: the beam may be sited at three positions on $P$ which is \emph{top}, \emph{mid} and \emph{bottom}.

Following our strategy on \emph{OAG}, the chase is starting at beam $\mathcal{B}$. By definition of MVC, a pure beam as self organizing obviously has been a MFS. However, beam $\mathcal{B}$ at either mid or top of sequence $P$, the operator should move from $\mathcal{B}$ and along to bottom or peer to peer. The course would be at minimum  cost to absorb nodes, then the result of component $\varg_i\in\varg$ should be a MSF.

When at the bottom of $P$, it is clear the operator merely does along peer to peer, the $\varg_i$ will still be a MSF. 

\end{proof}
~\newline
Certainly, by the lemma, the cut property of MST seems to give a sign that is to say while we recover the group $\varg$ of MSFs by the same strategy as it has been to nodes, by a similar way we should gain a MST. But there is a barrier, the merger course on $\varg$ may be an iterative process with many loops. The case could bring out a basis problem: given two subsets $\varg_i, \varg_j\in\varg$, by adding a minimum bridge that bestrides two trees, whether the new tree $\varg_i\oplus\varg_j$ is a MSF or MST? We shall use the below lemma to answer this question.

\begin{lemma}\label{l2}
It is reasonable to suppose the input instance $\varg$ as $\vert\varg\vert\geq 2$. After the method \emph{OAG} exercising on $\varg$, for an new flotilla $\varg^{\/*}=\varg_i\oplus\varg_j$ ($\varg_i, \varg_j\in \varg$), it is a MSF.
\end{lemma}

\begin{proof}
Given $\varg_i\cup \varg_j\subseteq\varg$. Assume an edge $\varepsilon$ incident to weight $\omega_{\varepsilon}$. We suppose edge $\varepsilon$ as a unique bridge to connect $\varg_i$ and $\varg_j$, and set the new form $\varg^{\/*}$ composed by $\varg_i, \varg_j$ and $\varepsilon$: $\varg^{\/*}=\varg_i\cup \varg_j\cup\varepsilon$. 

We firstly prove the form $\varg^{\/*}$ is a tree. Assume contrary there is at least a cycle on $\varg^{\/*}$. Given a nodes $u\in\varg_i$, if a traversal starts at $u$ and pass through edge $\varepsilon$ to reach another tree $\varg_j$, then is clear that there must be another bridge across the two trees $\varg_i$ and $\varg_j$, or else the traversal could be forced to pass through edge $\varepsilon$ repeatedly if to return back to node $u$. The conclusion is as such truth to swap positions of $\varg_i$ and $\varg_j$ in former proof.

We may let function $\mathcal{T}$ to figure the sum of weight of a tree such as $\mathcal{T}(\varg)=\sum_{e_i\in\varg}\omega_{e_i}$. We can let $\omega_{\varepsilon}$ is the minimum amid all bridges that are bestriding two trees. The total weight on tree new $\varg^{\/*}$ as:

\begin{flushleft}$\mathcal{T}(\varg^{\/*})=\mathcal{T}(\varg_i\cup\varg_j) + \omega_{\varepsilon}$.\end{flushleft}

Assume there is new MSF $\overline{\varg}$ distinct to tree $\varg^{\/*}$ and including all nodes that has been in two trees $\varg_i, \varg_j$. The total weight on $\overline{\varg}$ can be 
\begin{flushleft}$\mathcal{T}(\overline{\varg})=\mathcal{T}(\overline{\varg}^{\prime}) + \omega_{\varepsilon^{\prime}}$.\end{flushleft}

For two trees $\varg^{\/*}, \overline{\varg}$, it is obvious that the both edge or nodes quantities are same via they are tree. 

In above term, $\varepsilon^{\prime}$ can be set to a bridge, which two ends respectively are in $\varg_i$ and $\varg_j$ as well as the edge $\varepsilon$. By our premise about $\varg_i, \varg_j$ and $\varepsilon$, the total weight $\mathcal{T}(\overline{\varg}^{\prime})$ can but be $\geq\mathcal{T}(\varg_i\cup\varg_j)$; and furthermore there is $\omega_{\varepsilon^{\prime}}\geq\omega_{\varepsilon}$. Then we have conclusion $\mathcal{T}(\overline{\varg})\geq\mathcal{T}(\varg)$.

\end{proof}

The rest work is to take easy to comprehend similar to Seth\textquoteright s third stage\cite{8}. But there is a drag on our method that is about variable $counter$. By the $5^{th}$ feature above, the most quantity of clusters may be $counter =n/2$ where all of output clusters are including two nodes only. 

This detail of process can enforce us no choice to use an adjacent matrix with $n^{2}/2$ to store connected relation of all probability among these clusters. Consequently this drawback may induce the computing on cluster stage at complexity $O(n^2)$. In practice, it would bring out much more potential operations on memory, such as reinitializing, delete, update and etc to aggravate our implementation.

Therefore we need to reform OAG method. Moreover the breakthrough of the stalemate would be that the chase on node is not longer routinely on two directions and needing to record the relations of fleet model for any sorting. 

\subsection{Optimizing}
Given node $r$ and its own leaf set $L(r)$, underneath fleet model, we allow the subset $\mathcal{A}\subseteq L(r)$ consisted of components $\mathcal{J}, \mathcal{P}, \mathcal{S}$. The three components respectively contain those leaves sorting of \emph{boat}, \emph{beam} and \emph{towboat}. The compliment $L(r)\setminus\mathcal{A}$ will be trivial we take ignore-all.

In process of OAG method, the strike of chase is from $\mathcal{S}$ to $\mathcal{P},\mathcal{J}$ and stretches in $\mathcal{P}$. If this chase does at inverse strike, there is hereby emerged of invert pitfall which may damage our career.

If we view $\mathcal{A}$ as a stair then the component $\mathcal{P}$ likes a \emph{landing} amid the below $\mathcal{J}$ and up $\mathcal{S}$. By lemma\eqref{l2}, it gives us a surprise version that consider two MSFs and one\textquoteright s top side with MVCs connect to another\textquoteright s bottom. It is possible for them to merge together through a corridor on that band among two MSFs and, to become a new MSF. Then for set $\mathcal{A}$, components $\mathcal{P}$ and $\mathcal{S}$ with same MVC constructs a vicinal side of the below MSF. 

In fact, by strategy of OAG, the vicinal side at top of below MSF similarly is a rift to separate an original unified MSF, which cause is possibly made by algorithm oneself. We call this case \emph{fragmentality}.

Of course, we can be in another way to sort those candidates in an order of ascent. As though the minimum beam must be the top, but it is obvious that this patchwork can bring out a new cost maybe involving to $O(n^2)$ at most.

After we summarize the foregoing analysis about a form and its significance at a node, we can raise an \emph{inheritance system} to solve the invert pitfall appearing in the course of bottom to top. We describe the strategy in following.

\begin{enumerate}
\item We initially nominate any start node as an \emph{inheritor}. This title could be legal to trace upwards inside a flotilla and, may be imparted to another while a chasing task finishes at current inheritor.

\item In a practical process, an healthy inheritor must match the clause: $\mathcal{S}\cup\mathcal{P}\neq\varnothing$. Of course, the ill one is $\mathcal{S}\cup\mathcal{P}=\varnothing$. If no successor, in the next period of chase, there would be no upward again till chase over in current MSF. 

\item In common case, the direction of from top to bottom and peer to peer is the routine chase in a flotilla. 

\item By theorem\eqref{t1}, it knows the action of upward chase will stop at the top of flotilla. Therefore the inheritance system guarantees us to avoid the invert pitfall emerge.
\end{enumerate}
~\newline
We merely show the pseudo code about inheritance system in cluster stage:
\begin{enumerate}
\item To component $\varg_{i}=(c_i, \overline{\xi}_i)$ and $\varg_{i}\in\varg$, scans every node $r\in c_i$, then forwards to AWT over root $r$. Then after compares among those AWTs without \emph{arc} in $c_i$ to figure out MVC for each cluster $C_{\text{min}}[i]=\omega_{c_i}^{\ltimes}$.

\item To scan those above AWTs again. Then there are some several comparisons among native MVC $C_{\text{min}}[r]$, neighbors $C_{\text{min}}[\ell]$ and each bridge\textquoteright s incident weight $\omega_{r,\ell}$ to decide if to merge the objective cluster which is at the \emph{up}, \emph{horizon} and \emph{down} positions in graph $\varg$. This aggregation is an iterative progress by redo such operation till to no cluster for merger. Once a chasing done in any MSF $\varg_i$, we would write those nodes into a vehicle $V^{\/*}$ as a new cluster. That $V^{\/*}$ is temporary array with data structure of array $C$ alike. We can let the two arrays to swap memory address one another, so that they could alternatively work for store clusters. Of course, once an eligible arc is picked up, incident two endpoints would be record into graph $\overline{\xi}$ as having done in node stage.

\item The protected mechanism is same as in node stage: some nodes may be refreshed and possessed of a new cluster \emph{id} or else going on with the old one.
\end{enumerate}

It is obvious that the course of cluster merger is iterative within which for those nodes, some withdraw and some be left, so they would be treated as such again and again till there only a new cluster\textquoteright s born as the final result. That is the MST about given instance. On intuitive sense, the new measure should be lighter on fragmentality than OAG in node stage.

We only showed cluster stage, in effect, the node stage can be viewed of a cluster stage which every cluster is singleton which distinct nest incident to distinct bird. It is obvious for every iterative loop of glomeration that the runtime is $O(2E)$ at node stage. We call this method \emph{omni OAG}, abbr. \emph{oOAG}.

By $5^{th}$ feature, we can understand the worst case for conglomerating, it at most will go up to $\log n$ iterative loops that in each cluster merger, the new one whatever is just one-into-one; i.e. the each new cluster always contains two older ones, so that the flow work may be depicted into a binary tree alike. Hence the overall runtime complexity is $O(2E\log n)$ at worst case. 

Of course, there are other methods could aid us to speed up execution. In our process, we need to quicken the speed of producing and inquiry over MVCs for all nodes or clusters, then the \emph{Fabonacci heap}\cite{9} can acquire MVC in $O(1)$ by applying \emph{findmin} measure on those AWTs. 

The runtime will be reduced to $O((E+n)\log n)$ if we could omit the stage of producing MVC for each subgraph and cluster. But in process of merger, the relationship and MVC about those clusters are unceasingly changed of lots. Then the heap should be appropriately modified and tuned following the situation changing at each iterative loop. Since that, we consider $O((E+\log n)\log n)$ \cite{9} be more approach to practice for the worst case on updating Fabonacci heap

~\newline
\emph{A melioration.} A measure seems to refine the complexity. That is to cut graph in dynamical process which snips out those nodes which are not abutting to other clusters. I.e. root oneself with all of its own leaves together are staying in a same cluster. We settle an operator to filter out those nodes in the step of seeking MVC for every cluster. Since that, we can reduce every input of next loop. Consider the worst case on a complete graph that the size of each leaf set is $n-1$, we have a concise recursive function to figure the input for per loop.
\begin{flalign*}
\delta_{k+1}&=\delta_{k}-2^{k-1}n&\\
&\text{s.t.}\quad k\geq0;~\delta_{0}=E;~ \delta_{1}=E-n.&
\end{flalign*}
The sum is $\sum_{k\leq\log n}\delta_{k}=E\log n - o(n)$. For $o(n)\leq E$ then $E(\log n-1)=E\log(n/2)$ is the deterministic sum. But the improvement is \emph{extra limited}.

As though yielded such assertion, this measure reminds us that the survivors out of snipping in a loop, they are composed of the current \emph{vicinal edge} or \emph{peripheral points} for native cluster. Thus at different step of glomeration, by the measure, we can outline a surface of a new cluster but no necessary to occupy much more resources.

\subsection{The $k$ Value}
Sometimes for clustering, it is essential to firstly take the number of kernels of conglomerations. In above pseudo code, that is the variable $counter$ said of $k$ value in tradition. It can be solved of this problem by theorem\eqref{t1}: the kernel is a group of nodes all in beam possessed of $\mathcal{S}= \varnothing$ in each leaf set. This definition is naturally to involve your application, maybe you feel the pure beam with $\mathcal{J}\cup\mathcal{S}= \varnothing$ is yet as real kernel.

However, we develop a new method: the procedure will be implemented along the beam among those nodes we have said peer to peer. If meets any leaf set with $\vert\mathcal{S}\vert> 0$ then halts and jumps to another that have not been charged. The variable $counter$ will be increment to sum the amount of kernels by oneself adding one as it does in other methods. Hereby this subroutine runtime complexity still is $O(E)$ and should be settled in the process of making MVC in each local weight set without occupying more resources. And the kernels may be a group of MSFs with a data structure akin to graph $\varg$.

Moreover the consequence merger will readily be implemented from top to bottom which is completely to sweep the case of \emph{fragmentality}. This method at intuitive sense is similar to the \emph{density-based clustering algorithm} made by Ester, Martinet et al[1996]\cite{10}. But there is difference in the idea that by theorem\eqref{t1}. Since those MVCs in a kernel can share with a same number that may not be sure of the least one in whole data, just do in that flotilla, thereby the kernel is produced from some natural relation said of \emph{proximate} in cluster theory. That is certainly not to need to give a subjective number as the initial distance, but by the various \emph{resolutions} we can govern the producing of kernels. Finally we call this method \emph{kernel OAG}, abbr. kOAG.

\subsection{Experiment}
We select a lattice object to simulate an oOAG operation on an image. To given a pixel $u$ in latticing network, we supposed its neighbors are those suchlike sited at \emph{up, down, left, right} and four \emph{corners} on the virtual sides surround $u$, which construct a rectangle. Since that, a common root should own 8 neighbors near to oneself, unless those on edges of figure contiguous to ambient margin.

Our machine was a laptop with Intel I3 core, 4G memory and Win10 OS. The executive procedure was compiled by C++ at console platform. 

We let the number of points start at 1M $(n=10^{6})$ formed by square term $n=p^2$ where $p$ is the number of points at either of row and column on lattice. The number of 1M acts as benchmark $q_{0}$, that others $q_{i}$ is of $>q_{0}$ and the ratio $q_{i}/q_{0}$ as the sequence point at abscissa. We add the $p$ value with 500 by 500 till $p=3k$ ($k=10^{3}$), then the number of arcs is from round 8M to 72M. 

We use a set $Q=(1,2,\ldots,10)$, these nonzero natural numbers to assign those arcs as weight, there is a ratio of \emph{density} denoted by $\gamma = \vert Q\vert/\vert \tau\vert$. We give the figure 1 to show the practical runtimes on several levels of numbers.

\begin{center}
\includegraphics{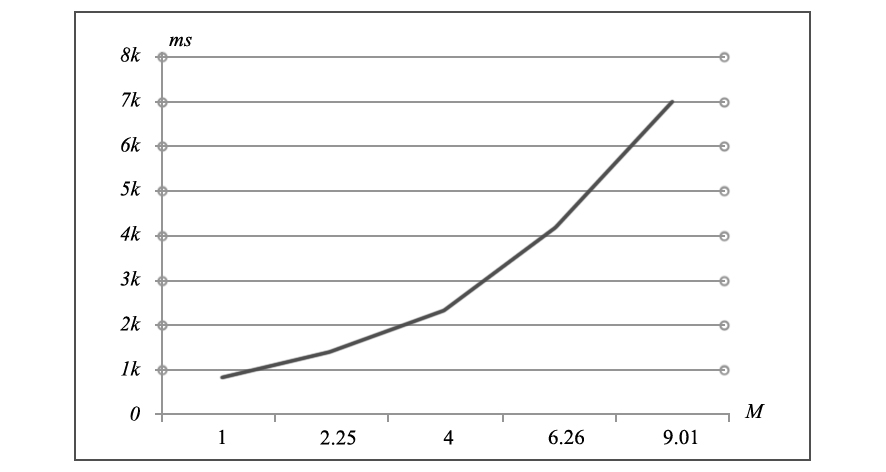}\\
Figure 1
\end{center} 

 The following table 1 shows some concrete numbers.

\begin{table}[h]\renewcommand{\arraystretch}{1.2}
\caption{\small{MST EXP. 1}}
~\newline
\centering
\begin{tabular}{|l | c|  r|}
\hline
$n=$\quad\quad\quad&\quad $k=$\quad\quad\quad&\quad \textbf{R.T.} (ms)\quad\quad\\
\hline
1\textbf{M}&\quad 1\,301&\quad 470\\
\hline
2.25\textbf{M}&\quad 1\,383&\quad 1\,110\\
\hline
4\textbf{M}&\quad 1\,323&\quad 2\,078\\
\hline
6.25\textbf{M}&\quad 1\,261&\quad 4\,093\\
\hline
9\textbf{M}&\quad 1\,256&\quad 7\,031\\
\hline
\end{tabular}
\end{table}

The results show we took less than 8 seconds to deal with approach 36M edges. Those $k$ values are showing the amount of clusters and, $k-1$ edges will be sought out at cluster stage. On the other hand, the five numbers are so closed that reflect the efficiency of aggregation is in increment with the scale of instance. 

Meanwhile we had a statistic analysis round this trial, it particularly refers to a phenomenon: the price, the mean of picking an eligible edge over many of arcs $A$ that had taken part in practical comparing. We use the ratio $n/A$ to indicate the efficiency of aggregation.

When sets the density $\gamma$ in a constant in trial and merely change number of points. The results show while $\vert V\vert$ come up to 1M, the ratio $n/A$ would close to 1. Another for $p=100$ the ratio is about $0.3$ almost even with $n/\vert\tau\vert$. Contrast to $k$ value, it in larger instance is less than small one. Such case is likewise coincide to the viewpoint or observation about efficiency of aggregation made by R. Tarjan\cite{5}. And the volume of density $\gamma$ actually is not distinctly apparent over to impact efficiency of aggregation. Certainly this assertion may be slapdash without many more data or theory to support.  

In fact, this is just a simple trial, but refers to a research way which in a stochastic system either absolute or approximate, it enables us to estimate the mean of weights on MST by those data.

\section{Summary}
Surely, our model cannot avoid the number $\log n$ of steps at the worst case. It is said of the optimizing work in today, justly our thought is not distinct to past for example the \emph{Fabonacci heap} or \emph{Soft heap} (was designed by Bernard Chazelle in 2000) and so forth, that attempted to reduce the effort of accessing to the data. The \cite{11} has done useful work at this aspect which by integrating and strengthening hardware and communication to optimize practical tactic.

Anyway, by integrating above methods, advanced memory technique and suitable deployments of computation, the overall complexity may be reduce to grade of $O(E)$. 

Besides the complexity, we showed an entire detail of course of a glomeration. This is importance to many applications as a generalized platform that means we can add many and many conditions on it to simulate a reality world. 

After all, MST problem is simple and basis as Seth\textquoteright s command\cite{8}, and may be as a key phase of solving other problems, into which they would be transformed and optimized.


\begin{thebibliography}{99}
\bibitem{1}Yong Tan. (2013). Construct Graph Logic. CoRR abs/1312.2209. url: http://arxiv.org/abs/1312.2209
\bibitem{2}Petrica C. Pop. (2002). The Generalized Minimum Spanning Tree Problem. Dutch: Twente University Press. 
\bibitem{3}JOSEPH B B. Kruskal. (1956). On the shortest spanning subtree of a graph and the traveling salesman problem. the American Mathematical Society 7, 48-50, .
\bibitem{4}R.C. Prim. (1957). Shortest connection networks and some generalizations, Bell Systems Technical Journal, 36, 1389-1401.
\bibitem{5}H. Gabow, T. Spencer \& R. Tarjan. (1986). Efficient algorithms for finding minimum spanning trees in undirected and directed graphs. Combinatorica, 6(2):109–122.
\bibitem{6}D. Karger, P. Klein \& R. Tarjan. (1995). A randomized lineartime algorithm to find minimum spanning trees. ACM, 42(2):321–328.
\bibitem{7}C. Zahn. (1971). Graph-theoretical methods for detecting and describing gestalt clusters. IEEE Transactions on Computers, C-20:68–86.
\bibitem{8}Pettie Seth, \& Ramachandran Vijaya. (2000). An Optimal Minimum Spanning Tree Algorithm. In: Montanari U., Rolim J.D.P., Welzl E. (eds) Automata, Languages and Programming. ICALP 2000. Lecture Notes in Computer Science, vol 1853. Springer, Berlin, Heidelberg
\bibitem{9}Michael L. Fredman, \& Robert Endre Tarjan. (1987). Fibonacci heaps and their uses in improved network optimization algorithms. Journal of the ACM (JACM) Volume 34 Issue 3, July 1987  Pages 596-615. 
\bibitem{10}Ester, Martin. Kriegel, Hans-Peter. Sander, Jörg. Xu, Xiaowei (1996). Simoudis, Evangelos. Han, Jiawei. Fayyad, Usama M., eds. A density-based algorithm for discovering clusters in large spatial databases with noise. the Second International Conference on Knowledge Discovery and Data Mining (KDD-96). AAAI Press. pp. 226–231.
\bibitem{11}Artem Mazeev, Alexander Semenov, Alexey Simonov. (2016). A Distributed Parallel Algorithm for Minimum Spanning Tree Problem. CoRR, abs/1610.04660. url: http://arxiv.org/abs/1610.04660
\end{thebibliography}
\end{document}